\documentclass[twocolumn,aps,pra,reprint,superscriptaddress,amsmath,amssymb,bbm,floatfix]{revtex4-2}
\usepackage[latin9]{inputenc}
\usepackage{color}
\usepackage{amstext}
\usepackage{graphicx}
\usepackage[colorlinks, allcolors={blue}]{hyperref}
\usepackage{physics}
\usepackage{blindtext}

\usepackage{xcolor,graphicx}
\usepackage{newlfont}
\usepackage{amssymb,amsmath,mathrsfs,amsthm}
\usepackage{verbatim}
\usepackage{bbm}
\usepackage{multirow}
\usepackage[varg]{txfonts}

\newcommand{\orcid}[1]{\href{https://orcid.org/#1}{\includegraphics[width=7pt]{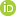}}}
\newcommand{\mc}{\mathcal}
\newcommand{\be}{\begin{equation}}
\newcommand{\ee}{\end{equation}}
\newcommand{\mbb}{\mathbbm}

\newtheorem{prop}{Proposition}

\renewcommand\bra[1]{{\langle{#1}|}}
\makeatletter
\renewcommand\ket[1]{%
  \@ifnextchar\bra{\k@t{#1}\!}{\k@t{#1}}%
}
\newcommand\k@t[1]{{|{#1}\rangle}}
\makeatother

\begin{document}

\title{High-dimensional monitoring and the emergence of realism via multiple observers}

\author{Alexandre C. Orthey Jr. \orcid{0000-0001-8111-3944}}
\affiliation{Center for Theoretical Physics, Polish Academy of Sciences, Al. Lotnik\'ow 32/46, 02-668 Warsaw, Poland.}
\affiliation{Institute of Fundamental Technological Research, Polish Academy of Sciences, Pawi\'nskiego 5B, 02-106 Warsaw, Poland.}

\author{Pedro R. Dieguez \orcid{0000-0002-8286-2645}}
\affiliation{International Centre for Theory of Quantum Technologies, University of Gda\'nsk, Jana Bazynskiego 8, 80-309 Gda\'nsk, Poland.}

\author{Owidiusz Makuta \orcid{0000-0002-0070-8709}}
\affiliation{Center for Theoretical Physics, Polish Academy of Sciences, Al. Lotnik\'ow 32/46, 02-668 Warsaw, Poland.}
\affiliation{Instituut-Lorentz, Universiteit Leiden, P.O. Box 9506, 2300 RA Leiden, The Netherlands}
\affiliation{$\langle \text{aQa}^\text{L} \rangle$ Applied Quantum Algorithms Leiden, The Netherlands}

\author{Remigiusz Augusiak \orcid{0000-0003-1154-6132}}
\affiliation{Center for Theoretical Physics, Polish Academy of Sciences, Al. Lotnik\'ow 32/46, 02-668 Warsaw, Poland.}

\begin{abstract}
Correlation is the basic mechanism of every measurement model, as one never accesses the measured system directly. Instead, correlations are created, codifying information about the measurable property into the environment. Here, we address the problem of the emergence of physical reality from the quantum world by introducing a model that interpolates between weak and strong non-selective measurements for qudits. By utilizing Heisenberg-Weyl operators, our model suggests that independently of the interaction intensity between the system and the environment, full information about the observable of interest can always be obtained by making the system interact with many environmental qudits, following a Quantum Darwinism (QD) framework.
\end{abstract}

\maketitle

\section{Introduction}

Quantum theory gives a prominent role to the notion of measurements. They are the basic ingredients in several quantum technologies such as measurement-based quantum computation~\cite{briegel2009measurement,raussendorf2016symmetry}, thermal devices fueled by measurements~\cite{behzadi2020quantum,buffoni2019quantum,bresque2021two,lisboa2022experimental,dieguez2023thermal}, measurement-based quantum communication~\cite{zwerger2016measurement}, as well as in foundational discussions regarding the measurement problem~\cite{schlosshauer2005decoherence,dieguez2018information} and on the understanding of the quantum-to-classical transition~\cite{zurek2003decoherence,zurek2009quantum}. In particular, the emergence of objective reality has been investigated with the framework of QD~\cite{zurek2009quantum,touil2022eavesdropping} through the process of redundancy, where multiple copies of information about the quantum system are created in its environment, and from the closely related Spectrum Broadcast Structure~\cite{korbicz2014objectivity,horodecki2015quantum,korbicz2021roads}. 

Generalized measurements that can interpolate between weak and strong (projective) non-selective regimes were employed to investigate the role of measurements in the emergence of realism from the quantum substratum~\cite{dieguez2018information}, as quantified by the informational measure known as \textit{quantum irrealism}~\cite{bilobran2015measure,dieguez2018information,mancino2018information}.
The quantum irrealism measure is based on the contextual realism hypothesis introduced in~\cite{bilobran2015measure} which generalizes the notion of EPR elements of reality~\cite{einstein1935can} by stating that for quantum systems, a measured property becomes well-defined after a projective measurement of some discrete spectrum observable, \textit{even when one does not have access to the specific measurement result}~\cite{bilobran2015measure,dieguez2018information,orthey2022quantum}. In other words, incoherent mixtures of all possible outcomes have realism for the measured observable. Realism was investigated employing monitoring with continuous variable measurement systems~\cite{dieguez2018information}, which showed to have a complementary relation with the available information of a quantum system~\cite{dieguez2018information}. The information-realism complementarity suggests that the establishment of realism for some observable is grounded on the encoding of information about it.

\begin{figure}[t]
    \centering
    \includegraphics[width=0.7\linewidth]{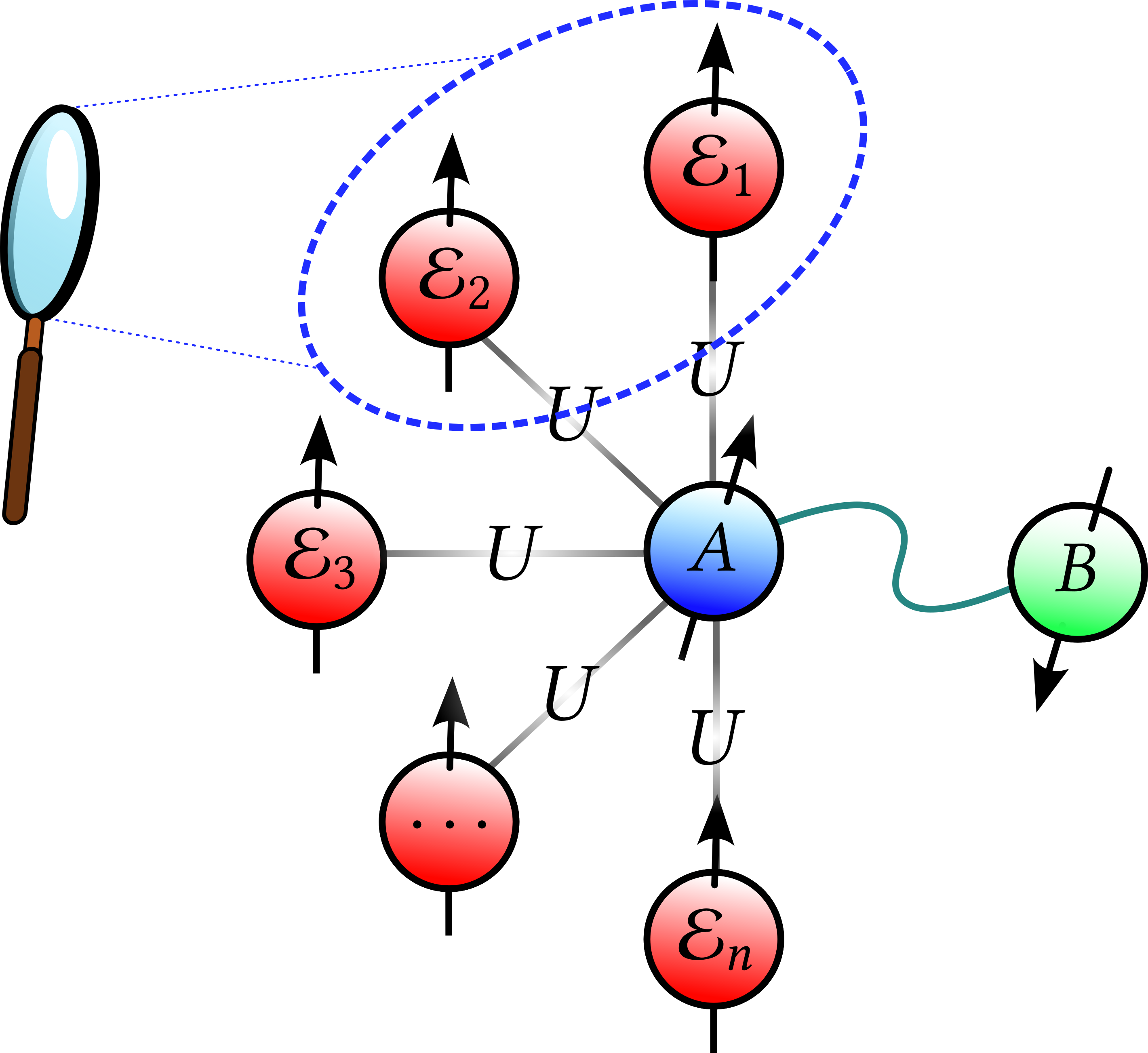}
    \caption{A bipartite system $\rho_\mc{AB}$ with arbitrary dimension is monitored by a collection of environmental subsystems $\{\mc{E}_i\}_{i=1}^n$, with each interaction represented by a unitary evolution acting on the joint system. The interaction with the environment establishes not just the realism associated with the context $\{A,\rho_\mc{AB}\}$, but also the proliferation of redundant information about the system in small portions of the environment; a key process in QD. This proliferation of redundant information is associated with the emergence of objective reality as several observers will agree on their outcomes after accessing the information about the system that is encoded in the environment.}
    \label{fig1}
\end{figure}

Here, we investigate the emergence of realism via monitoring modeled by discrete quantum systems with higher dimensions.  We identify that a large qudit-environmental system is sufficient for the establishment of realism for all the range of measurement strengths, corroborating the results obtained with continuous variable monitoring~\cite{dieguez2018information}. Moreover, we introduce a consistent interpolation model for weak to strong measurements within high-dimensional systems. First, we discuss the qubit regime and address its limitations and direct generalizations. To address the qudit regime, we model the interaction between the system and the environment through \textit{generalized observables}~\cite{kaniewski2019maximal}, described by the Fourier transform of POVMs. Our model allows for the control of the disturbance on the measured quantum system, therefore, allowing an interpolation from a weak to strong projective action~\cite{oreshkov2005weak,dieguez2018information}. This interpolation regime was experimentally investigated for qubits employing a trapped-ion platform~\cite{pan2020weak}, and with photonic weak measurements to investigate the information-realism complementarity~\cite{mancino2018information}. In the noiseless regime, our model reproduces the ``perfect record'' case from Zurek's work \cite{zurek2009quantum}. In Fig. \ref{fig1} we depict the scenario we are modeling.

The paper is organized as follows. In Sec. \ref{irrealism_sec} we review the concept of irrealism of observables introduced in \cite{bilobran2015measure}. In Sec. \ref{info_realism_app} we revisit the information-realism complementarity introduced by \cite{dieguez2018information}. In Sec. \ref{qubit_sec} we show how our problem can be solved for qubit systems. In Sec. \ref{generalized_sec} we review generalized observables in order to expand the solution for qubits to the qudit case. In Sec. \ref{results} we present the main results and final considerations are left for Sec. \ref{conclusion}.

\section{Preliminaries}

\subsection{Irrealism of observables}\label{irrealism_sec}

The quantum realism argument can be formalized as follows.  Let $A=\sum_a aA_a$ be a $d$-output observable acting on $\mc{H_A}$, where $A_a=\ket{a}\bra{a}$ are projectors satisfying $\sum_a A_a=\mbb{1}_\mc{A}$. Also, let 
\be
\Phi_A(\varrho)=\sum_a (A_a\otimes \mbb{1}_\mc{B})\varrho(A_a\otimes \mbb{1}_\mc{B})
\ee
be the non-selective projective measurement of the observable $A$. Let us consider a bipartite system $\rho_\mc{AB}$ acting over $\mc{H}_\mc{AB}$ and an observable $A$ acting on $\mc{H_A}$. We say that $\rho_\mc{AB}$ has a realism defined for $A$ \textit{iff}
\be\label{reality}
\Phi_A(\rho_\mc{AB})=\rho_\mc{AB}.
\ee
States that already have realism defined for some observable are invariant to a non-selective (or non-revealed) projective measurement of the same observable, i.e. $\Phi_A(\Phi_A(\rho))=\Phi_A(\rho)$. Subsequently, the \textit{irreality} of $A$ given $\rho$
\be\label{irreality}
\mathfrak{I}_A(\rho)\coloneqq \min_{\varrho} S\big(\rho||\Phi_A(\varrho)\big)=S\big(\Phi_A(\rho)\big)-S(\rho)
\ee
proved to be a faithful quantifier of $A$-realism violations for a given state $\rho$~\cite{bilobran2015measure,dieguez2018information,orthey2022quantum}, where $S(\rho||\sigma)=\Tr\big[\rho(\log_2{\rho}-\log_2{\sigma})\big]$ stands for the relative entropy and $S(\rho)\coloneqq-\text{Tr}\big(\rho\log_2{\rho}\big)$ is the von Neumann entropy. In addition, we can decompose irreality as local coherence plus non-optimized quantum discord, i.e. 
\be
\mathfrak{I}_A(\rho)=\mc{C}(\rho_\mc{A})+D_A(\rho)
\ee
where $\mc{C}(\rho_\mc{A})=\mathfrak{I}_A(\Tr_\mc{B}\rho)$, $D_A(\rho)=I_\mc{A:B}(\rho)-I_\mc{A:B}(\Phi_A(\rho))$, and $I_\mc{A:B}=S(\rho_\mc{A})+S(\rho_\mc{B})-S(\rho_\mc{AB})$ is the quantum mutual information. Consequently, correlations between parties $\mc{A}$ and $\mc{B}$ prevent the existence of elements of physical reality for observables on both parties, that is, any positive value for $D_{A(B)}(\rho)$ implies non-null irrealism $\mathfrak{I}_{A(B)}(\rho)$.

The irrealism measure was theoretically employed to study a series of foundational problems~\cite{gomes2018nonanomalous,Rudnicki2018uncertainty,orthey2019nonlocality,gomes2019resilience,fucci2019tripartite,freire2019quantifying,costa2020information,lustosa2020irrealism,orthey2022quantum,gomes2022realism,Engelbert2023considerations,paiva2023coherence,Caetano2024quantum,orthey2024geometric}. Moreover, some experimental reports include: a nuclear magnetic resonance experiment to probe the robustness of the wave-particle duality in a quantum-controlled interferometer~\cite{dieguez2022experimental}, photonic~\cite{mancino2018information} and superconducting qubits~\cite{basso2022reality} to investigate the emergence of realism upon \textit{monitoring maps}~\cite{dieguez2018information, dieguez2018weak}, and the role of distant operations in the erasure of physical reality in an optical quantum eraser experiment \cite{Araujo2024quantum}. Monitoring maps \cite{dieguez2018information} are defined as
\be\label{monitoring}
\mc{M}_A^\epsilon(\varrho)\coloneqq (1-\epsilon)\rho+\epsilon\Phi_A(\varrho),
\ee
that is, they describe an interaction that produces as an output the interpolation between weak and strong non-selective measurement regimes for the measured system. Noteworthy, \cite{Molitor2024Nov} analyzes the implementation of monitoring maps for incompatible observables in superpositions of causal orders, employing a quantum switch influenced by an environment modeled through a collisional approach.

\subsection{Information-realism complementarity}\label{info_realism_app}

Consider the amount of quantum accessible information of a generic quantum state $\rho$ in a Hilbert space $\cal{H_S}$ with dimension $d$ as
\begin{equation}
I(\rho)\coloneqq \ln{d}-S(\rho).
\end{equation}
Following Ref.~\cite{dieguez2018information}, wherein a complementarity relation between information and the degree of irrealism for some discrete observable was introduced, we can check that the subsequent application of $n$ monitoring maps over a state that is already a state of reality for $A$ does not change the state of affairs, that is
\begin{align}
[\cal{M}_A^{\epsilon}]^n\big(\Phi _{A}(\rho )\big)&
=(1-\epsilon )^{n}\Phi _{A}(\rho )+\big[1-(1-\epsilon )^{n}\big]\Phi
_{A}(\rho )  \nonumber \\
& =\Phi _{A}(\rho )=\Phi _{A}\big([\cal{M}_{A}^{\epsilon}]^n(\rho )\big).
\end{align}%
The above also proves the hierarchy of the map $\Phi _{A}$ over $[\cal{M}_A^{\epsilon}]^n$. Interestingly, this shows that the map $\cal{M}_{A}^{\epsilon}$ commutes with the map $\Phi_{A}$ for all intensities $\epsilon$. Employing entropy concavity  and the non-negativity of the irrealism measure $\mathfrak{I}_A(\rho)$, we can evaluate the difference in the irrealism given monitoring of the same observable
\begin{align}
\Delta \mathfrak{I}_A &\coloneqq \mathfrak{I}_A(\rho)- \mathfrak{I}_A(\cal{M}_{A}^{\epsilon}(\rho))\nonumber\\
&=S\left(\cal{M}_{A}^{\epsilon}(\rho )\right)-S(\rho)\geqslant \epsilon \mathfrak{I}_A(\rho),
\end{align}%
with equality for $\epsilon =0$. Introducing the amount of remaining accessible information after some monitoring of a generic observable $A$ as
\begin{equation}\label{info_realism}
\Delta I_{\cal{S}}\coloneqq I(\rho)-I(\cal{M}_{A}^{\epsilon}(\rho )),
\end{equation}
we see that
\be
\Delta I_{\cal{S}}=S\left(\cal{M}_{A}^{\epsilon}(\rho )\right)-S(\rho)=\Delta \mathfrak{I}_A
\ee
which means that the irrealism of the observable $A$ for the preparation $\rho$ quantifies the amount of remaining accessible information that still can be extracted after a weak non-selective measurement of $A$. It is important to note that either when $\epsilon=1$ or when $\lim_{n\to+\infty}[\cal{M}_A^{\epsilon}]^n$ for infinitely successive applications, $\mathfrak{I}_A(\rho)=\Delta I_{\cal{S}}$, which implies emergence of realism for this observable.

Moreover, one can prove that the information flow between the system $\cal{H_{S}}$ and the total environment $\cal{H_{E}}$ when we impose a global unitary dynamics that reproduces the monitoring of $A$ is
\be
\Delta\mathfrak{I}_A=\Delta( I_{\cal{S:E}}+I_{\cal E}), 
\ee
where $I_{\cal{S:E}}(\rho)=S(\rho_{\cal{S}})+S(\rho_{\cal{E}})-S(\rho)$ stands for the quantum mutual information. In other words, because the environment $\cal{E}$ gets information about $A$, this observable increases its realism degree. In the limit of strong non-selective measurements, irrealism goes to zero as well as the accessible information about this observable.

\subsection{Qubit case}\label{qubit_sec}

Let us describe the monitoring map from the point of view of the global interaction between the system and the environment. The simplest model describing a measurement procedure is constituted by a single CNOT gate. Suppose that our system of interest $\mc{S}$ is composed of only one arbitrary qubit in the state $\rho_\mc{S}$, but the environment $\mc{E}$ is known to be $\rho_\mc{E}=\ket{0}\bra{0}$. A CNOT gate~\cite{nielsen2010quantum}
\be
U_\textrm{CNOT}=\ket{0}\bra{0}\otimes\mbb{1}+\ket{1}\bra{1}\otimes\sigma_x
\ee
keeps the qubit in the environment in the same state if the state of the system is $\ket{0}$ and flips it if the state is $\ket{1}$. If the system lies in a superposition, such a superposition will be extended to the environment in the form of entanglement. Indeed, the global state $\Omega\coloneqq \rho_\mc{S}\otimes\ket{0}\bra{0}$ evolves to
\be\label{Omega_qubit}
\Omega'\coloneqq U\Omega U^\dagger=\sum_{i,j=0}^{1}\bra{i}\rho_\mc{S}\ket{j}\  \ket{i}\bra{j}\otimes\ket{i}\bra{j}.
\ee
If we discard the information in the environment by tracing it out, the resulting operation over the system is a non-selective measurement in the form of
\be
\Tr_\mc{E}\Omega'=\sum_{i\, =\, 0}^{1}\bra{i}\rho_\mc{S}\ket{i}\ \ket{i}\bra{i}=\Phi_{\sigma_z}(\rho_\mc{S}),
\ee
where $\sigma_z=\ket{0}\bra{0}-\ket{1}\bra{1}$. The map $\Phi_{\sigma_z}$ means that a measurement in the $z$ direction was realized, but the outcome was not recorded; a sufficient condition for the establishment of realism for observable $\sigma_z$. Since~\eqref{Omega_qubit} is symmetric, we have that $\Tr_\mc{S}\Omega'=\Tr_\mc{E}\Omega'$, and thus the probability $p_i=\Tr(\ket{i}\bra{i}\rho_\mc{S})$ is now encoded in the environment. In a practical situation, this is not always the case, as the CNOT gate cannot be perfectly implemented due to operational errors. Thus, only partial information regarding the observable can be retrieved; which, however, can be circumvented by interacting the system with more qubits in the environment to create redundant accessible information, following the QD framework.

One way to model a noisy CNOT gate for qubits was explored in~\cite{touil2022eavesdropping} by simply replacing $\sigma_x$ operator in $U_\textrm{CNOT}$ with the unitary $\sigma_{\theta}$, which is defined as a linear combination $\sigma_\theta\coloneqq\cos\theta\sigma_x+\sin\theta\sigma_z$. In this work, we also show that by controlling the interaction intensity between the system and the single environmental qubit through the parameter $\theta$, we can interpolate between weak and strong projective non-selective measurements. In detail, a noisy $U_\textrm{CNOT}$ with $\sigma_\theta$ results in a monitoring map $\mc{M}_A^\epsilon$ after tracing out the environment (see Appendix \ref{qubit_cases}). However, since this procedure only applies to qubit systems, it is likely to be not a good model for more complex systems. By restricting ourselves to only two dimensions we may miss certain effects which appear only for high-dimensional systems. Therefore, our motivation is to find such a model for arbitrary local dimensions. Consequently, we need to find a unitary operation acting on qudits that can be adjusted to interpolate between weak and strong non-selective measurement regimes.

\subsection{Generalized observables}\label{generalized_sec}

One way to generalize $U_\textrm{CNOT}$ for qudits (including the noisy regime) is to replace the Pauli operators with the Heisenberg-Weyl operators
\be
Z=\sum_{k\, =\, 0}^{d-1}\omega^k\ket{k}\bra{k},\qquad X=\sum_{k\, =\, 0}^{d-1}\ket{k+1}\bra{k},
\ee
where $\ket{d}\equiv\ket{0}$, $\omega=e^{2\pi\mbb{i}/d}$ is the complex root of unity, and $d=\dim\mc{H}$ is the dimension of the Hilbert space. It can be verified that $X^\dagger X=Z^\dagger Z=\mbb{1}_\mc{H}$, $X^d=Z^d=\mbb{1}_\mc{H}$, and $Z^l X^k=\omega^{lk}X^kZ^l$. Also, when $d=2$, then $Z$ and $X$ are the usual $\sigma_z$ and $\sigma_x$ Pauli matrices, respectively. Thus, $\sigma_\theta$ becomes $\cos\theta X+\sin\theta Z\eqqcolon \overline{\sigma}_\theta$. Because we have $d$ outcomes, we need to take the powers of $\overline{\sigma}_\theta$ to codify the information about the outcome in different states of the environment:
\be\label{U_noisy}
U_\textrm{noisy CNOT}=\ket{0}\bra{0}\otimes\mbb{1}+\ket{1}\bra{1}\otimes \overline{\sigma}_\theta+\ket{2}\bra{2}\otimes \overline{\sigma}_\theta^2+\ldots.
\ee
However, the above operator cannot always be used because it does not represent a unitary evolution in general. The reason why is because $\overline{\sigma}_\theta$ is not unitary for $d\geqslant 3$. In fact, it is impossible to obtain a unitary operator as a linear combination of only $X$ and $Z$ for $d\geqslant 3$~\cite{kaniewski2019maximal,santos2022scalable}.  Nonetheless, it is known that one can always find a linear combination of operators in the form $\{XZ^k\}$ or $\{ZX^k\}$ for $k\, =\, 0,\ldots,d-1$ that results in a unitary operator~\cite{bandyopadhyay2002new}.

Although $Z$ contains complex eigenvalues, it can fully characterize a valid observable. More generally, consider a $d$-outcome quantum measurement $M=\{M_a\}$ which is defined by positive-semidefinite operators $M_a$ acting on a Hilbert space $\mc{H}$, such that $\dim\mc{H}=d$ and $\sum_a M_a=\mbb{1}_\mc{H}$. The measurement $M$ is simply a POVM, but if $M_aM_{a'}=\delta_{aa'}M_a$, then $M$ is a projective measurement. For qudits, it is convenient to use generalized observables by taking a discrete Fourier transform of $M$ given by
\be
T^{(i)}=\sum_{a\, =\, 0}^{d-1} \omega^{ia}M_a.
\ee
Immediately, $T^{(0)}=\mbb{1}_\mc{H}$. By performing the inverse Fourier transform on the generalized observables $T^{(i)}$, one can recover the measurement $M$ with 
\be
M_a=\dfrac{1}{d}\sum_{i\, =\, 0}^{d-1}\omega^{-ia}T^{(i)}.
\ee
Therefore, the generalized observables $T^{(i)}$ fully characterize the measurement $M$. Also, the POVM $\{M_a\}$ is a projective measurement \textit{iff} $T\coloneqq T^{(1)}$ is a unitary matrix. In that case, $T$ is called the \textit{unitary observable} of $\{M_a\}$ and $T^{(i)}=T^i$~\cite{kaniewski2019maximal}.

The goal now is to find a unitary observable $T$, in replacement to $\overline{\sigma}_\theta$, that appropriately codifies the information regarding a given observable of interest into the environment, such that, the resulting operation over the system is a monitoring map~\eqref{monitoring}.

\section{Results}\label{results}

Let $\mc{H=H_S}\otimes\mc{H_E}$, where $\mc{S}$ and $\mc{E}$ stand for system and environment, respectively, such that $\mc{H_S}=\mc{H_A\otimes H_B}$ and $\mc{H_E}=\bigotimes_{k=1}^n \mc{H}_{\mc{E}_k}$, with party $\mc{A}$ containing one qudit and party $\mc{E}$ containing $n$ qudits (see Fig.~\ref{fig1}). The fragment $\mc{F}_m$ of the environment $\mc{E}$ is constituted by $\mc{H}_{\mc{F}_m}=\bigotimes_{k=1}^m \mc{H}_{\mc{E}_k}$, with $ m\leqslant n$. For our purposes, $\dim\mc{H_A}=\dim\mc{H}_{\mc{E}_i}\eqqcolon d$ for all $i$. 

We propose a noisy CNOT gate $U_{\mc{SE}_i}$ that acts over $\mc{H}_{\mc{AE}_i}$, which correlates one qudit of the system with qudit $i$ in the environment, and it is defined as
\be\label{U}
U_{\mc{SE}_i}\coloneqq\sum_{j\, =\, 0}^{d-1}P_j\otimes T^j,
\ee
where $P_j=A_j\otimes \mbb{1}_\mc{B}$ is the projector acting on subspace $\mc{H}_{\mc{S}}$ and $T$ is an operator given by
\be\label{operatorT}
T\coloneqq\sum_{k\, =\, 0}^{d-1}\alpha_kZX^k,
\ee
where $\alpha_k$ are coefficients such that $\sum_{k=0}^{d-1}|\alpha_k|^2=1$. For $T$ to be a unitary observable (see previous section), the following conditions must hold
\be
\alpha_k = \dfrac{1}{d}\sum_{l\, =\, 0}^{d-1}\omega^{lk}e^{\mbb{i}\phi_l},\quad\text{ and }\quad
\sum_{l=0}^{d-1} \phi_l=0,
\ee
where, now, $T$ is specified by a set of $d-1$ phases $\{\phi_l \}$ (see Appendix \ref{properties}). Operator $T$ does not depend on the choice of $A$ since $A$ is an arbitrary measurement and $T$ is just a way to encode information about the system in the environment. We have chosen the set $\{ZX^k\}$ instead of $\{XZ^k\}$ because we need non-zero terms in the main diagonal of our operator.

Without loss of generality, we can suppose that every qudit in the environment is in the state $\ket{0}_{\mc{E}_i}$. If one chooses to use operators $\{XZ^k\}$, it would require to prepare every environmental qudit in a Fourier transform of $\ket{0}$ in order to obtain the monitoring map. Therefore, we proceed with $\{ZX^k\}$. Let us begin with the interaction between the system and the first qudit in the environment given by $\rho_{\mc{SE}_1}'\coloneqq U_{\mc{SE}_1} \rho_{\mc{SE}_1} U_{\mc{SE}_1}^\dagger$, where $\rho_{\mc{SE}_1}=\rho_\mc{AB}\otimes\ket{0}\bra{0}_{\mc{E}_1}$. From~\eqref{U}, we get
\be\label{trace_E1}
\Tr_{\mc{E}_1}\rho_{\mc{SE}_1}'=\sum_{i,j\, =\, 0}^{d-1}P_i\:\rho_\mc{AB}\: P_j\:\bra{0} (T^j)^\dagger T^i \ket{0}_{\mc{E}_1}.
\ee
After some algebra (see Appendix \ref{properties}), it is possible to verify that
\be\label{braket0TT0}
\bra{0}(T^j)^\dagger T^i\ket{0}_{\mc{E}_1} = \dfrac{1}{d}\sum_{q=0}^{d-1}\exp\left(\mbb{i}\sum_{m=0}^{[i-j]_d-1}\phi_{[q+m]_d} \right),
\ee
where $[a]_d\coloneqq a\mod d$, for any $a\in\mathbb{Z}$. In particular, if $[i-j]_d=0$ in~\eqref{braket0TT0}, then $\bra{0}(T^j)^\dagger T^i\ket{0}_{\mc{E}_1}=1$.

A closer look at the monitoring map~\eqref{monitoring} allows us to write
\be
\mc{M}_A^{\epsilon=1-\eta}(\rho_\mc{AB})=\sum_{i,j=0}^{d-1}P_{i}\:\rho_\mc{AB}\: P_{j} \left[\eta+(1-\eta)\delta_{i,j} \right],
\ee
where $\eta\in[0,1]$ acts as the noise of the monitoring, i.e., the complement of the monitoring intensity. Since we require that 
\be\label{traceE1}
\Tr_{\mc{E}_1}\rho_{\mc{SE}_1}'=\mc{M}_A^{\epsilon=1-\eta}(\rho_\mc{AB}),
\ee
we can conclude from Eqs.~\eqref{trace_E1} to~\eqref{traceE1} that
\be\label{eta}
\eta=\dfrac{1}{d}\sum_{q=0}^{d-1}\cos\phi_{q},
\ee
such that, the phases $\{\phi_q\}_{q=0}^{d-1}$ must satisfy:
\begin{subequations}\label{system}
    \begin{align}
        &\eta\geqslant 0;\\
        &\sum_{q=0}^{d-1}\sin\left(\sum_{m=0}^{p}\phi_{[q+m]_d} \right)=0,\\
        &\sum_{q=0}^{d-1}\cos\left(\sum_{m=0}^{p}\phi_{[q+m]_d} \right)=\sum_{q=0}^{d-1}\cos\left(\sum_{m=0}^{p'}\phi_{[q+m]_d} \right),
    \end{align}
\end{subequations}
for every $p,p'$ satisfying $0\leqslant p<p'\leqslant d-2$ (see Appendix \ref{properties}). Although solving the above transcendental equations in full generality turns out to be difficult, it is possible to find a particular analytical solution for any $d$ given by $\phi_0=\theta$, $\phi_1=-\theta$, and $\phi_{k\geqslant 2}=0$. This solution for the particular case of $d=3$ gives rise to a noise intensity $\eta=(1+2\cos\theta)/3$. In order to find other and more nontrivial solutions, one can employ numerics. In fact, in our numerical exploration for dimensions $d=4,5,6,7$ such solutions exist for any noise rate $\eta\in[0,1]$, which are distinct from the analytical one we just mentioned.

Operator~\eqref{U} is indeed quite general. Because of the Stinespring dilation theorem~\cite{kretschmann2008,nielsen2010quantum}, the unitary operator that provides the monitoring map $\mc{M}_A^\epsilon$, after the partial trace, must be unique, up to local isometries. For all practical purposes, we have $\eta>0$, since we are always going to find the presence of noise in any interaction. Therefore, even if the access to the information present in the environment by a macroscopic observer occurs through a strong projective measurement, such an act does not extract the complete information regarding the observable of interest in the system. However, in the noiseless regime when $\eta=0$, our model reproduces the perfect record case $\ket{k}_\mc{S}\otimes\ket{0}_{\mc{E}_i}\to \ket{k}_\mc{S}\otimes\ket{k}_{\mc{E}_i}$ from Zurek's work \cite{zurek2009quantum}, as we are going to see further.

The natural question to ask now is how to codify full information about an observable ($\epsilon\to 1$) into the environmental qudits by using gate~\eqref{U}? We are going to show that to obtain full information about the observable we need a bigger environment, i.e. more qudits. Consider the global unitary that performs the interaction between the system with the entire environment $U\coloneqq U_{\mc{SE}_n}\ldots U_{\mc{SE}_1}$. The initial global state of the system+environment is given by $\Omega\coloneqq \rho_{\mc{AB}}\otimes\bigotimes_{i=1}^{n}\ket{0}\bra{0}_{\mc{E}_i}$. After the unitary evolution $\Omega'\coloneqq U \Omega U^\dagger$, we can write
\be
    \Omega' = \sum_{i,j\, =\, 0}^{d-1} P_i\: \rho\: P_j \otimes \bigotimes_{k=1}^n T_i\ket{0}\bra{0}_{\mc{E}_k} (T_j)^\dagger.
\ee
If we trace out the whole environment, the remaining state of the system will be given by (see Appendix \ref{U_and_Tr})
\be\label{traceE}
\Tr_{\mc{E}}\Omega'=\mc{M}_A^{1-\eta^n}(\rho_\mc{AB}).
\ee
Now, we can apply Eq. (18) from~\cite{dieguez2018information}, i.e.,
    \be\label{recursion}
    \left[\mc{M}_A^\epsilon\right]^n(\rho_\mc{AB})=\mc{M}_A^{1-(1-\epsilon)^n}(\rho_\mc{AB})
    \ee
    to Eq.~\eqref{traceE} to obtain
    \be\label{Mpowers}
    \Tr_{\mc{E}}\Omega' = \left[\mc{M}_A^{\epsilon=1-\eta} \right]^n(\rho_\mc{AB}).
    \ee
Eq.~\eqref{traceE} shows that subsequent interactions between the system and $n$ environmental qudits through the noisy CNOT gate~\eqref{U} result in a monitoring map of intensity $1-\eta^n$, which allows us to interpret $\eta^n$ as the \textit{effective noise}. From~\eqref{Mpowers}, we can see that each weak interaction with some part of the environment produces one more monitoring over the system. This process clearly decreases the degree of irreality~\eqref{irreality} and makes the system more classical regarding observable $A$~\cite{costa2020information,orthey2022quantum}.

Immediately, from~\eqref{traceE} and the definition of a monitoring, we have
\be\label{limit_phiA}
\lim_{n\to+\infty}\left(\Tr_\mc{E} \Omega' \right)=\Phi_A(\rho_\mc{AB}),
\ee
for any solution of system~\eqref{system}, apart from the trivial solutions that result in $\eta=1$ (full noise). The above relation holds for every dimension $d\geqslant 2$, every state $\rho_\mc{S}$, and every observable $A$, which means that the whole effective environment has now access to the full information about the observable $A$ regarding $\rho_\mc{S}$, as demonstrated by 
\be
I(\varrho)-I(\Phi_{A}(\varrho))=\mathfrak{I}_A(\varrho).
\ee
Thus,~\eqref{limit_phiA} means that to guarantee the codification of the total accessible quantum information regarding observable $A$ of any system into the environment, it is sufficient for the system to be monitored by a large environment. In other words, our work demonstrates that for high-dimensional monitoring, redundant information spread in a large environment is a sufficient criterion for the emergence of realism, as we highlighted in~\eqref{info_realism} by employing the information-realism complementarity.

Indeed, when the system reaches a reality state for $A$ (states s.t. $\mathfrak{I}_A(\rho)=0$), all the quantum informational content that one could extract about $A$ has already been extracted. However, if the initial state is already a state of reality for $A$, then no quantum information can be acquired, since $\Phi_A(\Phi_A(\varrho))=\Phi_A(\varrho)$. In that case, a measurement of $A$ is just a revelation of a pre-defined quantity. Moreover, when $\bra{0}\left(T^j\right)^\dagger T^i\ket{0}=\delta_{ij}$, this is, in the noiseless case, it is possible to relabel the state $T^i\ket{0}=\ket{i}$, and the states $\ket{i}$ form a basis. As such, we find that the interaction chosen is typical of what one expects from \.{Z}urek's works \cite{zurek2009quantum}
\begin{equation}\label{zurek}
    \ket{k}_\mc{S}\otimes\ket{0}_{\mc{E}_i}\to \ket{k}_\mc{S}\otimes\ket{k}_{\mc{E}_i}. 
\end{equation}
In what follows we provide an explicit example for $d=3$ of how our model results in \eqref{zurek}. In the qutrit case, Eqs. \eqref{system} give $\phi_0=\theta$, $\phi_1=-\theta$, and $\phi_2=0$, such that $\eta=(1+2\cos\theta)/3$. The noiseless case $\eta=0$ requires $\theta=2\pi/3$, which gives $T=\omega ZX^2$. We can use the fact that $Z^l X^k=\omega^{lk}X^kZ^l$ to obtain
\begin{equation}
    U_\mc{SE}=\ket{0}\bra{0}\otimes\mathbbm{1}+\ket{1}\bra{1}\otimes \omega ZX^2+\ket{2}\bra{2}\otimes Z^2X,
\end{equation}
where $T^0=\mathbbm{1}$, $T^1=\omega ZX^2$, and $T^2=Z^2X$. Immediately,
\begin{equation}
    U_\mc{SE}\ket{k}\otimes\ket{e_0}=\ket{k}\otimes T^k\ket{e_0}=\ket{k}\otimes \ket{e_k},
\end{equation}
where $\braket{e_k}{e_{k'}}=\delta_{kk'}$, which represents the ``perfect record'' case from \cite{zurek2009quantum}. It is important to mention that the POVMs to be implemented in the laboratory are given by the inverse Fourier transform of $T$ given by Eq.~\eqref{operatorT}.

\section{Conclusion}\label{conclusion}

Irrealism can be framed as a resource~\cite{costa2020information} and has an intimate connection with the concept of information~\cite{dieguez2018information,orthey2022quantum}. Monitoring is identified as a realistic operation, meaning that in general the irrealism of the context $\{A,\rho_\mc{AB}\}$ is destroyed by $A$-monitorings~\cite{dieguez2018information} via the flow of information between the system and the auxiliary system (that we called environment). We filled a gap in such connections by providing an explicit unitary interaction between qudits to derive interpolations from weak to strong non-selective measurements. The most important conclusion of our analysis is that, for high-dimensional discrete monitoring, a state of reality for a given observable can be reached regardless of the intensity of the interaction between the system and the environment, as long as the latter is composed of a sufficiently large number of particles. In addition, from \eqref{eta}, the noise rate $\eta$ is still highly dependent on parameters $\{\phi_q\}$ even when $d\to\infty$. For instance, if $\min_q\cos\phi_q=Q$, for some $Q>0$, then $\eta> Q$.

Following the QD framework, if the system manages to disseminate redundant information about an observable to many particles in the environment, such a system will move towards a state of objectivity. In a nutshell, equation~\eqref{limit_phiA} means that objectivity from QD implies $A$-reality states from quantum realism for discrete and high-dimensional monitoring. Other connections can be made by exploring this new relation, for instance, it would be worth analyzing the volume of system solutions~\eqref{system} for different values of $\eta$. That could allow us to determine, whether it is easier to find unitaries that lead to weak or, instead, strong monitoring maps, and how the dimension of the system influences that process. Still, another possibility for further research is the time needed to probe the system via interaction with the environment. In addition, one could also explore the continuous limit of our model and investigate the connection between quantum Darwinism and non-Markovianity similarly to what is done in \cite{oliveira2019quantum}.

\begin{acknowledgements}
We are grateful to Jaros{\l}aw Korbicz for useful discussions. A.C.O and R.A acknowledge the QuantERA II Programme (VERIqTAS project) that has received funding from the European Union's Horizon 2020 research and innovation programme under Grant Agreement No 101017733 and from the Polish National Science Center (projects No 2021/03/Y/ST2/00175 and No. 2022/46/E/ST2/00115). P.R.D. acknowledges support from the Foundation for Polish Science (FNP) (IRAP project, ICTQT, Contract No MAB/2018/5, co-financed by EU within Smart Growth Operational Programme). O.M. and R.A. acknowledge the Polish NSC through the SONATA BIS (grant No 2019/34/E/ST2/00369).
\end{acknowledgements}

\bibliography{bibliography.bib}

\appendix

\section{Qubit case --- the c-maybe gate}\label{qubit_cases}

Let us consider the c-maybe gate \cite{touil2022eavesdropping}
\be
U_\oslash\coloneqq \ket{0}\bra{0}\otimes\mbb{1}+\ket{1}\bra{1}\otimes\sigma_\theta=\sum_{k=0}^{1} P_k\otimes \sigma_\theta^k,
\ee
where $\sigma_\theta\coloneqq \cos\theta\sigma_x+\sin\theta\sigma_z$ acts over $\mc{H}_{\mc{E}_1}$ and $P_k=A_k\otimes\mbb{1}_\mc{B}$ are projectors in $\mc{H}_\mc{S=AB}$ s.t. $A=\sum_k a_k A_k$ is the observable of interest in $\mc{H_A}$, $A_kA_{k'}=\delta_{kk'}A_k$, and $\sum_k A_k=\mbb{1}_\mc{A}$. It can be checked that 
\be
\bra{0}\sigma_\theta^i \sigma_\theta^j\ket{0}=\sin\theta+\delta_{i,j}(1-\sin\theta),
\ee
for $i,j=0,1$. If the initial state of the global system is given by two qubits in the state $\rho_{\mc{S}\mc{E}_1}=\rho_\mc{AB}\otimes\ket{0}\bra{0}_{\mc{E}_1}$, then the evolved state is
\be
\rho_{\mc{S}\mc{E}_1}'\coloneqq U_\oslash \rho_{\mc{S}\mc{E}_1} U_\oslash^\dagger=\sum_{i,j=0}^{1} P_i\rho_\mc{AB} P_j\otimes \sigma_\theta^i\ket{0}\bra{0}_{\mc{E}_1}(\sigma_\theta^j )^\dagger.
\ee
If we trace out the environmental qubit, we obtain
\begin{align}
        \Tr_{\mc{E}_1} \left(\rho_{\mc{S}\mc{E}_1}'\right) &=\sum_{i,j=0}^{1} P_i\rho_\mc{AB} P_j \bra{0} (\sigma_\theta^j)^\dagger \sigma_\theta^i \ket{0}_{\mc{E}_1} \nonumber\\
        &= \sum_{i,j=0}^{1} P_i\rho_\mc{AB} P_j [\sin\theta+\delta_{i,j}(1-\sin\theta)],
\end{align}
which can be written as
\begin{align}
 \Tr_{\mc{E}_1} \left(\rho_{\mc{S}\mc{E}_1}'\right) &= \sin\theta\sum_{i,j=0}^{1} P_i \rho_\mc{AB} P_j +(1-\sin\theta) \sum_{i=0}^1 P_i \rho_\mc{AB} P_i \\ \nonumber& = \sin\theta\rho_\mc{AB} + (1-\sin\theta) \Phi_A (\rho_\mc{AB}). 
\end{align}
By the definition of a monitoring map
\be\label{monitoring_SM}
\mc{M}_A^\epsilon(\varrho)\coloneqq (1-\epsilon)\rho+\epsilon\Phi_A(\varrho),
\ee
we have that
\be
\Tr_{\mc{E}_1} \left(\rho_{\mc{S}\mc{E}_1}'\right) = \mc{M}_A^{1-\sin\theta}(\rho_\mc{AB}).
\ee
As we can see here, the c-maybe operator written in the eigenbasis of $A$ is precisely the unitary evolution, up to local isometries, that results in a monitoring of $A$ with intensity $\epsilon=1-\sin\theta$ after we trace out the environment.

For comparison, let us consider the noisy CNOT gate $U_{\mc{SE}_i}$ given by Eq. \eqref{U} for dimension $d=2$, i.e. $U_{\mc{SE}_i}=\ket{0}\bra{0}\otimes\mbb{1}+\ket{1}\bra{1}\otimes T$,
where $T =-\sin\theta\sigma_y+\cos\theta\sigma_z$ is obtained from Eq. \eqref{operatorT}. The resulting monitoring map will have intensity $\epsilon=1-\cos\theta$, reproducing the effects of the c-maybe gate from the work of Touil et al. \cite{touil2022eavesdropping}, up to a $\pi/2$ phase on $\theta$.

\section{Properties of operator $T$}\label{properties}

One identity of great importance is the following:
\be\label{identity_delta}
\sum_{k\, =\, 0}^{d-1}\omega^{(a-b)k}=d\delta_{a,b},
\ee
where $a$ and $b$ are any integers and 
\be
\delta_{a,b}\coloneqq \begin{cases}
1 &\textrm{for } a = b \mod d;\\
0 &\textrm{for } a \neq b \mod d.
\end{cases}
\ee

Now, consider the unitary observable $T$ given by
\be\label{operatorT_SM}
T\coloneqq\sum_{k\, =\, 0}^{d-1}\alpha_kZX^k,\qquad \alpha_k=\dfrac{1}{d}\sum_{l\, =\, 0}^{d-1}\omega^{lk}e^{\mbb{i}\phi_l},
\ee
with condition $\sum_{l=0}^{d-1}\phi_l=0$. Let us prove three essential properties of the operator $T$.
\begin{prop}\label{prop_unitary}
The operator $T$ is unitary.
\end{prop}
\begin{proof}
To make the calculations simpler, we can apply a unitary transformation
\be
V\coloneqq \dfrac{1}{\sqrt{d}}\sum_{i,j}\omega^{ij}\ket{i}\bra{j}
\ee
to $T$ such that $VTV^\dagger\eqqcolon \overline{T}$. Indeed, $T$ is unitary \textit{iff}
\be\label{overlineT}
\overline{T}=\sum_{k=0}^{d-1}\alpha_k XZ^{-k}
\ee
is unitary.

Explicitly,
\be
\overline{T}\overline{T}^\dagger=\left(\sum_{k\, =\, 0}^{d-1}\alpha_k ZX^k \right)\left(\sum_{q\, =\, 0}^{d-1}\alpha_q^*X^{-q}Z^{-1} \right).
\ee
By performing the products, we obtain
\be\label{overlineTTdagger}
\overline{T}\overline{T}^\dagger= \sum_{k\, =\, 0}^{d-1}|\alpha_k|^2\mbb{1}+\sum_{\substack{k,q\, =\, 0 \\ k\neq q}}^{d-1}\alpha_k\alpha_q^* XZ^{q-k}X^{-1}.
\ee
Let us start by developing the first sum in the r.h.s. of the above expression by using identity \eqref{identity_delta}:
\begin{align}\label{sum_alpha}
\begin{split}
\sum_{k\, =\, 0}^{d-1}|\alpha_k|^2&=\sum_{k\, =\, 0}^{d-1}\dfrac{1}{d^2}\sum_{k=0}^{d-1}\sum_{l,l'=0}^{d-1} \omega^{k(l'-l)}e^{\mbb{i}(\phi_l-\phi_{l'})}\\ &=\dfrac{1}{d}\sum_{l,l'=0}^{d-1}\delta_{ll'}e^{\mbb{i}(\phi_l-\phi_{l'})}=1.
\end{split}
\end{align}
For simplicity, let us denote $n\coloneqq q-k$ and rewrite \eqref{overlineTTdagger} in the following form
\be\label{TT1}
\overline{T}\overline{T}^\dagger=\mbb{1}+\sum_{n=1}^{d-1}\sum_{q\, =\, 0}^{d-1}\alpha_{q-n}\alpha_q^*XZ^nX^{-1}.
\ee
The remaining sum in the r.h.s. of \eqref{TT1} will vanish. Indeed, by implementing identity \eqref{identity_delta} multiple times, we obtain
\be
\sum_{q\, =\, 0}^{d-1} \alpha_{q-n}\alpha_q^*=\dfrac{1}{d^2}\sum_{q\, =\, 0}^{d-1} \sum_{m_1,m_2=0}^{d-1} \omega^{q(m_2-m_1)}e^{m_1n}e^{\mbb{i}(\phi_{m_1}-\phi_{m_2})}=0.
\ee
\end{proof}

\begin{prop}\label{prop_expectation}
Given an integer $j$, the operator $T$ satisfies
\be
\bra{0} T^j\ket{0}_{\mc{E}_i}= \dfrac{1}{d}\sum_{q=0}^{d-1}\exp\left(\mbb{i}\sum_{m=0}^{[j]_d -1} \phi_{[q+m]_d} \right),
\ee
where $[j]_d$ represents $j\mod d$. In particular, if $j=0$ or $j=d$, then $\bra{0} T^j\ket{0}_{\mc{E}_i}=1$.
\end{prop}
\begin{proof}
From definition \eqref{operatorT_SM}, we have
\be
T^j=\left(\sum_{k\, =\, 0}^{d-1} \alpha_k ZX^k \right)^j.
\ee
Since $X^kZ=\omega^{-k} ZX^k$, we can rewrite the r.h.s. of the above equation as
\be\label{Tj}
T^j=Z^j\prod_{l\, =\, 0}^{j-1}\left(\sum_{k_l\, =\, 0}^{d-1}\alpha_{k_l}\omega^{-lk_l}X^{k_l} \right).
\ee
Now, we can suppose that $j>1$ and take expected value $\bra{0}T^j\ket{0}_{\mc{E}_i}$. After taking the sum and the product, the r.h.s will be a sum of terms consisting of a complex number multiplying $X^{\kappa}$, where
\be
\kappa \coloneqq  \sum_{l\, =\, 0}^{j-1}k_l\, .
\ee
The only terms that will have a contribution to $\bra{0}T^j\ket{0}_{\mc{E}_i}$ are the terms for which 
\be\label{sum0}
\kappa  =\, 0
\ee
which results in  
\be\label{sum_prod_omega_alpha}
\bra{0}T^j\ket{0}_{\mc{E}_i}=\sum_{\substack{k_0,\ldots,k_{j-1}\, =\, 0\\ \kappa\, =\, 0}}^{d-1}\left(\prod_{l\, =\, 0}^{j-1}\omega^{-lk_l}\alpha_{k_l}\right).
\ee
By condition \eqref{sum0}, we can write
\be\label{expected_1}
\bra{0}T^j\ket{0}_{\mc{E}_i}=\sum_{k_1,\ldots,k_{j-1}\, =\, 0}^{d-1}\left(\alpha_{-(k_1+\ldots+k_{j-1})}\prod_{l=1}^{j-1}\omega^{-lk_l}\alpha_{k_l}\right).
\ee

From \eqref{operatorT_SM}, we have:
\be\label{alpha_k_open}
\alpha_{k_l}=\dfrac{1}{d}\sum_{q'=0}^{d-1}\omega^{q'k_l}e^{\mbb{i}\phi_{q'}};
\ee
\be
\alpha_{-(k_1+\ldots+k_{j-1})}=\dfrac{1}{d}\sum_{q=0}^{d-1}\omega^{-q(k_1+\ldots+k_{j-1})}e^{\mbb{i}\phi_q}.
\ee
By implementing the above equations to \eqref{expected_1} we obtain
    \begin{align}
        &\bra{0}T^j\ket{0}_{\mc{E}_i}\nonumber\\
        &= \sum_{k_1,\ldots, k_{j-1}=0}^{d-1} \dfrac{1}{d} \sum_{q=0}^{d-1} \Bigg[\omega^{-q(k_1+\ldots+ k_{j-1})}e^{\mbb{i}\phi_q} \prod_{l=1}^{j-1}\left(\dfrac{1}{d} \sum_{q'=0}^{d-1}\omega^{k_l(q'-l)}e^{\mbb{i}\phi_{q'}}\right)\Bigg] \notag\\
        &= \dfrac{1}{d^j}\sum_{k_1,\ldots, k_{j-1}=0}^{d-1} \sum_{q=0}^{d-1} \left[e^{\mbb{i}\phi_q}\prod_{l=1}^{j-1}\left(\sum_{q'=0}^{d-1}\omega^{k_l(q'-q-l)}e^{\mbb{i}\phi_{q'}}\right)\right].
    \end{align}
From identity \eqref{identity_delta}, we obtain
    \begin{align}
        \bra{0}T^j\ket{0}_{\mc{E}_i} &= \dfrac{1}{d}\sum_{q=0}^{d-1} e^{\mbb{i}\phi_q}\prod_{l=1}^{j-1}\sum_{q'=0}^{d-1}\delta_{0,q'-q-l}e^{\mbb{i}\phi_{q'}},\nonumber\\
        &= \dfrac{1}{d}\sum_{q=0}^{d-1}e^{\mbb{i}\phi_q}\prod_{l=1}^{j-1}e^{\mbb{i}\phi_{q+l}}.
    \end{align}
After we develop the product in the above equation, we obtain the desired result
\be
\bra{0}T^j\ket{0}_{\mc{E}_i} = \dfrac{1}{d}\sum_{q=0}^{d-1} \exp\left(\mbb{i}\sum_{m=0}^{[j]_d-1}\phi_{[q+m]_d} \right).
\ee

\end{proof}

\begin{prop}
    The operator $T$ satisfies $T^d=\mbb{1}$.
\end{prop}
\begin{proof}
Analogously to the proof of Proposition \ref{prop_unitary}, we can first prove that $\overline{T}^d=\mbb{1}$. Indeed, $\overline{T}^d=\mbb{1}$ \textit{iff} $T^d=\mbb{1}$.

Analogously to the previous proof, we can take the power of \eqref{overlineT} by doing
\be
   \overline{T}^j =\left(\sum_{k=0}\alpha_k XZ^{-k} \right)^j = X^j\prod_{l=0}^{j-1}\left(\sum_{k_l=0}^{d-1}\alpha_{k_l}\omega^{-lk_l}Z^{k_l} \right).
\ee
If we apply $j=d$ to the above equation and express $Z$ explicitly, we obtain
\be
\overline{T}^d=\prod_{l=0}^{d-1}\sum_{k_l=0}^{d-1}\dfrac{1}{d}\sum_{q=0}^{d-1}\sum_{m=0}^{d-1}\omega^{k_l(q+m-l)}e^{\mbb{i}\phi_q}\ket{m}\bra{m}.
\ee
From identity \eqref{identity_delta}, we have
\begin{align}
\overline{T}^d&=\prod_{l=0}^{d-1}\sum_{m=0}^{d-1}e^{\mbb{i}\phi_{l-m}}\ket{m}\bra{m}\nonumber\\
&=\sum_{m=0}^{d-1}\exp\left(\mbb{i}\sum_{l=0}^{d-1}\phi_{[l-m]_d}\right)\ket{m}\bra{m}=\sum_{m=0}^{d-1}\ket{m}\bra{m}=\mbb{1},
\end{align}
where we have used the fact that $\sum_{l=0}^{d-1}\phi_{l}=0$.

\end{proof}

\section{Unitary evolution and partial traces}\label{U_and_Tr}

Let us prove the following results.
\begin{prop}
Let $\rho_{\mc{SE}_1}\coloneqq \rho_{\mc{AB}}\otimes\ket{0}\bra{0}_{\mc{E}_1}$ be the initial state of the system+environment, $U_{\mc{SE}_i}$ be the noisy CNOT gate defined as
\be\label{U_SM}
U_{\mc{SE}_i}\coloneqq\sum_{j\, =\, 0}^{d-1}P_j\otimes T^j,
\ee
where $P_j=A_j\otimes \mbb{1}_\mc{B}$ is the projector acting on subspace $\mc{H}_{\mc{S}}$ and $T$ is the unitary observable given by \eqref{operatorT_SM}. Also, let $\mc{M}_A^{\epsilon}$ be the monitoring map \eqref{monitoring_SM} and the evolved state be $\rho_{\mc{SE}_1}'=U_{\mc{SE}_1}\rho_{\mc{SE}_1} U_{\mc{SE}_1}^\dagger$. If the unitary $U_{\mc{SE}_i}$ must result in the monitoring map, i.e.
\be
\Tr_{\mc{E}_1}\left(\rho_{\mc{SE}_1}'\right)=\mc{M}_A^{\epsilon=1-\eta}(\rho_\mc{AB}),
\ee
then 
\be
\eta=\dfrac{1}{d}\sum_{q=0}^{d-1}\cos\phi_{q},
\ee
such that
\begin{subequations}
    \begin{align}
        &\eta\in[0,1],\\
        &\sum_{q=0}^{d-1}\sin\left(\sum_{m=0}^{p}\phi_{[q+m]_d} \right)=0,\\
        &\sum_{q=0}^{d-1}\cos\left(\sum_{m=0}^{p}\phi_{[q+m]_d} \right)=\sum_{q=0}^{d-1}\cos\left(\sum_{m=0}^{p'}\phi_{[q+m]_d} \right),\quad\forall p\neq p'.
    \end{align}
\end{subequations}
for every $p,p'$ satisfying $0\leqslant p<p'\leqslant d-2$.
\end{prop}
\begin{proof}
From Eq. \eqref{U_SM}, we can write
\be
\rho_{\mc{S}\mc{E}_1}'=\sum_{i,j\, =\, 0}^{d-1} P_i\rho_{\mc{AB}} P_j\otimes T^i\ket{0}\bra{0}_{\mc{E}_1} (T^j)^\dagger.
\ee
By tracing out the fragment $\mc{E}_1$ of the environment (which, in this case, comprises the whole environment), we obtain
\be
    \Tr_{\mc{E}_1} \left(\rho_{\mc{S}\mc{E}_1}'\right) =\sum_{i,j\, =\, 0}^{d-1}P_i\rho_\mc{AB} P_j \bra{0}(T^j)^\dagger T^i\ket{0}_{\mc{E}_1}.
\ee
A closer look at the monitoring map \eqref{monitoring_SM} allows us to write it in the following way
\be
\mc{M}_A^{\epsilon=1-\eta}=\sum_{i,j=0}^{d-1}P_{i}\rho_\mc{AB} P_{j} \left[\eta+(1-\eta)\delta_{i,j} \right],
\ee
where $\eta\in[0,1]$ acts as the noise of the monitoring, i.e., the complement of the monitoring intensity. Because we require that 
\be\label{traceE1_appendix}
\Tr_{\mc{E}_1}\rho_{\mc{SE}_1}'=\mc{M}_A^{\epsilon=1-\eta}(\rho_\mc{AB}),
\ee
then we must have
\be\label{braket=eta}
\bra{0}(T^j)^\dagger T^i\ket{0}_{\mc{E}_1}=\eta+(1-\eta)\delta_{i,j}.
\ee
The above constraint implies that the braket on its l.h.s. must be a real non-negative number. Therefore, from Proposition \ref{prop_expectation}, the set of phases $\{\phi_j\}_{j=0}^{d-1}$ must satisfy
\be
\sum_{q=0}^{d-1}\sin\left(\sum_{m=0}^{p}\phi_{[q+m]_d} \right)=0,\quad\forall p\in\{0,\ldots,d-2\},
\ee
to make $\bra{0}(T^j)^\dagger T^i\ket{0}_{\mc{E}_1}$ a real number. In addition, if $[i-j]_d\neq 0$, then the braket in \eqref{braket=eta} must always result in the same number $\eta$, independently of $i$ and $j$. Thus, from Proposition \ref{prop_expectation}, we also must have that
\be
\sum_{q=0}^{d-1}\cos\left(\sum_{m=0}^{p}\phi_{[q+m]_d} \right)=\sum_{q=0}^{d-1}\cos\left(\sum_{m=0}^{p'}\phi_{[q+m]_d} \right),\quad\forall p\neq p',
\ee
which implies that
\be
\eta=\dfrac{1}{d}\sum_{q=0}^{d-1}\cos\phi_{q},
\ee
where the phases must be s.t. $\sum_{q=0}^{d-1}\phi_q=0$ and $\eta\in[0,1]$.

\end{proof}
Now, let us see what happens when we make the system interact with $n$ environmental qudits.
\begin{prop}
    Let $\Omega=\rho_{\mc{AB}}\otimes\bigotimes_{k=1}^{n}\ket{0}\bra{0}_{\mc{E}_k}$ be the initial state of the system+environment,  $U_{\mc{SE}_i}$ be the noisy CNOT gate defined in Eq. \eqref{U_SM}, $U\coloneqq U_{\mc{SE}_n}\ldots U_{\mc{SE}_1}$ be the global unitary, and $\mc{M}_A^{\epsilon}$ be the monitoring map defined in Eq. \eqref{monitoring_SM}. If $\Omega'=U\Omega_0 U^\dagger$, then
    \be
    \Tr_{\mc{E}}\left(\Omega'\right)=\mc{M}_A^{\epsilon=1-\eta^n}(\rho_\mc{AB}).
    \ee
\end{prop}
\begin{proof}
    First, let us use definition \eqref{U_SM} and rewrite the evolved global state in the following way:
    \begin{widetext}
    \begin{align}
    \Omega' &=U_{\mc{SE}_n}\ldots U_{\mc{SE}_1} \left(\rho_{\mc{AB}} \bigotimes_{k=1}^{n} \ket{0}\bra{0}_{\mc{E}_k}\right) U_{\mc{SE}_1}^\dagger\ldots U_{\mc{SE}_n}^\dagger \\
    & = U_{\mc{SE}_n}\ldots U_{\mc{SE}_2}\left[ \sum_{i,j\, =\, 0}^{d-1}\left(P_i\rho_{\mc{AB}}P_j\otimes T^i_{\mc{E}_1}\ket{0}\bra{0}_{\mc{E}_1}(T^j_{\mc{E}_1})^\dagger \right) \bigotimes_{k=2}^{n} \ket{0}\bra{0}_{\mc{E}_k}\right] U_{\mc{SE}_2}^\dagger\ldots U_{\mc{SE}_n}^\dagger.
    \end{align}
    Note that, in the above, we explicitly specified the space where each operator $T$ is going to act. Since $P_i$ are projectors, we can do the following:
    \begin{align}
           \Omega' & = U_{\mc{SE}_n}\ldots U_{\mc{SE}_3}\left[ \sum_{i,j,i',j'\, =\, 0}^{d-1}\left(P_{i'}P_i\rho_{\mc{AB}}P_jP_{j'}\otimes T^i_{\mc{E}_1}\ket{0}\bra{0}_{\mc{E}_1}(T^j_{\mc{E}_1})^\dagger\otimes T^{i'}_{\mc{E}_2}\ket{0}\bra{0}_{\mc{E}_2}(T^{j'}_{\mc{E}_2})^\dagger\right) \bigotimes_{k=3}^{n} \ket{0}\bra{0}_{\mc{E}_k}\right] U_{\mc{SE}_3}^\dagger\ldots U_{\mc{SE}_n}^\dagger \\ 
           & = U_{\mc{SE}_n}\ldots U_{\mc{SE}_3}\left[ \sum_{i,j\, =\, 0}^{d-1}\left(P_i\rho_{\mc{AB}}P_j\otimes T^i_{\mc{E}_1}\ket{0}\bra{0}_{\mc{E}_1}(T^j_{\mc{E}_1})^\dagger\otimes T^{i}_{\mc{E}_2}\ket{0}\bra{0}_{\mc{E}_2}(T^{j}_{\mc{E}_2})^\dagger \right) \bigotimes_{k=3}^{n} \ket{0}\bra{0}_{\mc{E}_k}\right] U_{\mc{SE}_3}^\dagger\ldots U_{\mc{SE}_n}^\dagger.
    \end{align}
    \end{widetext}
    By applying this procedure to all the unitaries, we can obtain
    \be
    \Omega'=\sum_{i,j\, =\, 0}^{d-1}\left(P_i\rho_{\mc{AB}}P_j\otimes\bigotimes_{k=1}^{n}T^i_{\mc{E}_k}\ket{0}\bra{0}_{\mc{E}_k}(T^j_{\mc{E}_k})^\dagger\right).
    \ee
    The above equation is the global (possibly entangled) state that represents the situation found at the end of the experiment depicted in Fig. \ref{fig1}. By tracing out all the $n$ environmental qudits $\mc{E}$, we can use Eq. \eqref{braket=eta} to obtain
    \be
    \Tr_{\mc{E}}\Omega' = \sum_{i,j\, =\, 0}^{d-1} P_i \rho_{\mc{AB}} P_j \left[\eta+\delta_{i,j}(1-\eta) \right]^n.
    \ee
    Now, let us separate the above sum into two parts:
    \be
    \Tr_\mc{E}\Omega' =\sum_{i\, =\, 0}^{d-1} P_i \rho_\mc{AB} P_i +\eta^n\sum_{\substack{i,j\, =\, 0\\ i\neq j}}^{d-1} P_i \rho_\mc{AB} P_j.
    \ee
    Now, let us sum and subtract the term $\eta^n \sum_{i\, =\, 0}^{d-1} P_i \rho_\mc{AB} P_i$ to write the expression as a combination of $\Phi_A$ maps
    \begin{align}
    \Tr_\mc{E}\Omega' &= \sum_{i\, =\, 0}^{d-1} P_i \rho_\mc{AB} P_i +\eta^n \sum_{i\, =\, 0}^{d-1} P_i \rho_\mc{AB} P_i -\eta^n \sum_{i\, =\, 0}^{d-1} P_i \rho_{\mc{AB}} P_i\nonumber\\
    &+\eta^n\sum_{\substack{i,j\, =\, 0  \\ i\neq j}}^{d-1} P_i \rho_\mc{AB} P_j\\
    &= (1-\eta^n)\sum_{i\, =\, 0}^{d-1} P_i \rho_\mc{AB} P_i+\eta^n \sum_{i,j=0}^{d-1}P_i \rho_\mc{AB} P_j  \\
    &= (1-\eta^n)\Phi_A(\rho_\mc{AB})+\eta^n\rho_\mc{AB}.\label{1-eta^n}
    \end{align}
    From the definition of a monitoring map \eqref{monitoring_SM}, we see that the above expression is indeed a monitoring of intensity $\epsilon=1-\eta^n$,
    \be
        \Tr_{\mc{E}}\left(\Omega'\right)=\mc{M}_A^{\epsilon=1-\eta^n}(\rho_\mc{AB}).
    \ee

\end{proof}

\end{document}